\documentclass[aps,pra,longbibliography,superscriptaddress,nofootinbib,twocolumn,10pt]{revtex4-1}

\usepackage[T1]{fontenc}
\usepackage[utf8]{inputenc}

\usepackage{graphicx}
\usepackage[table]{xcolor}
\usepackage{comment}

\usepackage[english]{babel}

\usepackage{times}

\usepackage{amsmath}
\usepackage{bbm}
\usepackage{amsfonts}
\usepackage{amssymb}
\usepackage{amsthm}
\usepackage{mathbbol}
\usepackage{mathtools}
\usepackage{complexity}
  \newclass{\StoqMA}{StoqMA}
 
\usepackage{soul}

\usepackage{enumerate}

\PassOptionsToPackage{pdftex,hyperfootnotes=false,pdfpagelabels}{hyperref}
\usepackage{hyperref}
\hypersetup{%
    colorlinks=true, linktocpage=true, pdfstartpage=3, pdfstartview=FitV,%
    breaklinks=true, pdfpagemode=UseNone, pageanchor=true, pdfpagemode=UseOutlines,%
    plainpages=false, bookmarksnumbered, bookmarksopen=true, bookmarksopenlevel=1,%
    hypertexnames=true, pdfhighlight=/O,
    urlcolor=blue, linkcolor=blue, citecolor=green!60!black, 
}

  \theoremstyle{plain}
  \newtheorem{theorem}{Theorem}

  \theoremstyle{definition}
  \newtheorem{definition}[theorem]{Definition}

  \newtheorem{conjecture}{Conjecture}
  
  \theoremstyle{remark}
  
  \theoremstyle{plain}
  \newtheorem*{theorem*}{Theorem}
  \newtheorem*{lemma*}{Lemma}
  \newtheorem*{corollary*}{Corollary}
  \newtheorem*{proposition*}{Proposition}
  \newtheorem*{claim*}{Claim}
  \newtheorem*{problem*}{Problem}
  
\newcommand{\ii}{\mathbb{I}}

\newcommand{\norm}[1]{\left\| #1 \right\|}
\newcommand{\bra}[1]{\langle #1 \vert}

\newcommand{\ket}[1]{\vert #1 \rangle}

\newcommand{\fu}{Dahlem Center for Complex Quantum Systems, Physics Department, Freie Universit\"{a}t Berlin, Germany}
\newcommand{\fum}{Department of Mathematics and Computer Science, Freie Universit\"{a}t Berlin, Germany}

\definecolor{dominik}{RGB}{237,16,118}
\definecolor{daniel}{RGB}{0,0,200}

\usepackage{ulem} 
\makeatletter 
 \hypersetup{pdftitle = {Pinned QMA: The power of fixing a few qubits in proofs},
       pdfauthor = { Daniel Nagaj, Dominik Hangleiter, },
       pdfsubject = {quantum complexity},
       pdfkeywords = {quantum complexity,
       ground state connectivity, 
       quantum zeno effect, 
       universality,
       QMA
          }
      }
\makeatother

\begin{document}
\title{Pinned QMA: The power of fixing a few qubits in proofs}

\author{Daniel Nagaj}
\email[Corresponding author: ]{daniel.nagaj@savba.sk}
\affiliation{RCQI, Institute of Physics, Slovak Academy of Sciences, Bratislava, Slovakia}

\author{Dominik\ Hangleiter}
\affiliation{\fu}

\author{Jens Eisert}
\affiliation{\fu} \affiliation{\fum}

\author{Martin Schwarz}
\affiliation{\fu}
 
\begin{abstract}
What could happen if we pinned a single qubit of a system and fixed it in a particular state? First, we show that this leads to difficult static questions about ground state properties of local Hamiltonian problems with restricted types of terms. In particular, we show that the Pinned Commuting and Pinned Stoquastic Local Hamiltonian problems are QMA-complete. Second, we investigate pinned dynamics and demonstrate that fixing a single qubit via often repeated measurements results in universal quantum computation with commuting Hamiltonians. Finally, we discuss variants of the Ground State Connectivity problem in light of pinning, and show that Stoquastic GSCON is QCMA-complete. 


\end{abstract}

\maketitle

\section{Introduction}

The goal of quantum Hamiltonian complexity \cite{HamComplex,HamiltonianComplexity} 
is to study the computational power of physical models described by local 
Hamiltonians, the intricate properties of their dynamics and their eigenstates, as well as to understand the computational complexity of determining these properties.
Many Hamiltonians are known to be universal for quantum computation \cite{universal}, while others are thought to be much simpler, but still hard to investigate classically \cite{bravyihastingsTIM} or even efficiently simulable by classical computation \cite{matchKempe}.
There is a long history of searching for the simplest possible, closest to realistically and efficiently implementable, and robustly controllable interaction with universal dynamics for quantum computation with local Hamiltonians. Restrictions on the type and strength of interactions, locality, and geometrical restrictions have been investigated, e.g., 
in Refs.~\cite{KKR06, universal2body, OT06, universal, GossetTerhalParallel, hqca1D}.
Thinking about universality for computation often comes hand in hand with asking complexity questions such as identifying 
the hardness of determining the properties of the eigenstates of these Hamiltonians.

Looking at this from a quantum control theory viewpoint provides us with an interesting observation. An extra level of control over a subsystem can result in a boost in state generation possibilities, or the difficulty of complexity questions. We have seen this with the DQC1 (``one clean qubit'') model \cite{dqc1,dqc1morimae}, whose single fully initializable (clean) qubit gives rise to quantum advantage over classical computation. Similarly, if one is allowed to use magic states, computing with a restricted set of universal gates such as Clifford gates \cite{MagicStatesBravyi} becomes universal for quantum computation. Effectively fixing parts of the system to a particular state using perturbation gadgets allowed us to build complex effective Hamiltonians from simpler ones \cite{PhysRevA.77.062329}. It has also been shown that a Zeno-effect measurement of a small subsystem can grant universal power to a non-universal set of commuting gates \cite{burgarth}. 

In this work, we investigate the computational potential offered by controlling a small subsystem.
We focus on a specific type of control called {\em pinning} -- fixing the state of a small subsystem. 
Orsucci et al.\ \cite{purification} have formulated the related question of \emph{Hamiltonian purification}, investigating universal dynamics for a set of commuting Hamiltonians, projected into a particular subspace. 
As often in Hamiltonian complexity, there are two views of this task, a static and a 
dynamic one.
Our goal is to uncover in which situations pinning-induced effective interaction terms (weighted sums of the original restricted terms) lead to an increase in complexity, or state preparation power. 
 In both approaches, we prove several results complementing what we know about the hardness of problems without the special control.

First, statically, we ask about the difficulty of finding the properties of low-energy states of pinned Hamiltonians.
We pin a qubit by an external prescription and show that determining the lowest energy in the pinned subspace
 is \QMA-complete for a variety of restricted classes: commuting, stoquastic, Markov, and permutation Hamiltonians.  
With this we wish to shed light on the complexity of these problems without pinning, which we believe to be weaker. One of these is the currently actively investigated \emph{Commuting Local Hamiltonian} problem, which can be \NP-complete, or have ground states with topological order \cite{SchuchCommuting, aharonov_complexity_2018}. 
At the same time, we know that the ground state connectivity problem for commuting Hamiltonians is \QCMA-complete \cite{GMV}.
Another is the \emph{Stoquastic Local Hamiltonian} Problem, whose complexity is in the class \StoqMA\ which contains \NP\ and \MA, but is strongly believed weaker than \QMA\ \cite{bravyihastingsTIM}.
Taking the method of pinning to the extreme, we finally show that it yields \QMA-hardness results for Hamiltonians that are as simple as permutation matrices.
Some of our results on determining ground state energies of pinned problems complement the conclusions of Ref.~\cite{JGL} involving energy of the highest excited state. 

Second, dynamically, asking about the preparation power of evolution with restricted time-independent Hamiltonians combined with Zeno pinning of a qubit, we find connections to previous work on Hamiltonian purification \cite{burgarth,purification}, showing that the quantum Zeno-effect can drive \emph{efficient universal quantum computation} in several restricted settings. This includes, in particular, commuting Hamiltonians. 
Thanks to the details of the constructions, our results carry time and space requirements/guarantees from universal evolution models with unrestricted Hamiltonians.

Third, we find an application of pinning for the \emph{Ground State Connectivity (GSCON)} 
problem \cite{GharibianSikora} and its variants with restricted types of terms. 
Specifically, we prove that \emph{GSCON} with stoquastic Hamiltonians is \emph{\QCMA-complete}, complementing a similar recent result on \emph{Commuting GSCON} \cite{GMV}. 

Finally, we note that there are strong limits to the pinning technique. 
First, dimensionality arguments from Ref.~\cite{purification} mean a necessary increase in the size of the purified system in which interactions are restricted. 
Second, we encounter questions regarding locality of the required terms. 
Note that pinning does not allow us to create multiplicative effective terms, as perturbative gadgets do -- creating effective $3$-local terms from $2$-local ones. We do not know if it is possible to build gadgets for effective $k+1$ local interactions from $k$-local Hamiltonian terms with the help of pinning, for our commuting or stoquastic settings. Many such questions with low locality thus remain open.

Doing something special on a single additional qubit is not new. Besides Ref.\ \cite{GMV}, where the idea has been exploited to show that the \emph{GSCON}  problem is \QCMA-complete already for commuting Hamiltonians,
Jordan, Gosset and Love \cite{JGL} have used techniques 
tracing back to Ref.~\cite{JanWoc} to get rid of varying signs of matrix elements by increasing the system size and replacing positive 1's by 2$\times$2-identity matrices and negative $1$'s by the Pauli $X$ matrices.
They prove universality of adiabatic quantum computation in an excited state of a \emph{Stoquastic Local Hamiltonian}, 
instead of the usual ground state computation, by splitting the Hilbert space into two, depending on the state of an auxiliary qubit. Moreover, adding a stoquastic term effectively pinning this auxiliary into a state that results in a high energy, they showed \emph{\QMA}-completeness of understanding energy bounds for the highest excited energy of a Stoquastic Local Hamiltonian. \emph{Stoquastic Local Hamiltonians}
are those local spin Hamiltonians whose matrix elements in the standard basis satisfy the condition that all off-diagonal matrix elements are real and non-positive
\cite{bravyi_complexity_2009,bravyi_complexity_2006,marvian_computational_2018,Curing,Easing}. 
Next, they also show \emph{\QMA} hardness of bounding the lowest energy of doubly stochastic (Markov) matrices, and \emph{\QMA}$_1$ hardness of the \emph{Stochastic $6$-SAT problem} (deciding whether a sum of stochastic matrices is frustration-free or not).

This work is structured as follows: First, in Section~\ref{sec:static}, we show that several restricted versions of the \emph{Pinned Local Hamiltonian} problem are \emph{\QMA}-complete, in particular, commuting, stoquastic and permutation Hamiltonians. 
In Section~\ref{sec:dynpin} we then turn to the dynamical problem of universal time evolution, showing that the Zeno-pinned time evolution under both commuting and stoquastic Hamiltonians is complete for universal quantum computations. 
Finally, in Section~\ref{sec:GSCON}, we prove that the stoquastic \emph{GSCON} problem is \QCMA-complete and discuss the \emph{free fermionic GSCON} problem.

\section{Pinned Local Hamiltonians: a complexity viewpoint}
\label{sec:static}
\subsection{Local Hamiltonians and states with fixed qubits}

In \emph{\QMA}, a verifier asks for a witness of the form $\ket{\psi}$, to which she adds a few auxiliary qubits and verifies it with a quantum circuit $V$. Does anything change, if she demands that the witness must have a few qubits that are \emph{pinned} to some fixed state? No, as the verifier can ask for all but the pinned qubits of the witness, supply those pinned qubits on her own, and verify the whole state as before.

Rather straightforwardly, we can show that problems in the class \emph{\QMA} can be verified using \emph{Pinned \QMA} and vice versa, so that 
\emph{Pinned \QMA} $=$ 
\emph{\QMA}. If we ask for a {\em pinned} proof of the form $\ket{\psi'} = \ket{\psi}\ket{0}$, with one pinned qubit, the extra demand does not increase the complexity of the problem. If the verifier that asks for $\ket{\psi'}$ is $V'$, the same thing can be verified in \emph{\QMA} with a modified circuit $V$ which adds one more auxiliary system that stores a check of whether the pinned qubit is really $\ket{0}$, and then does the verification $V'$, accepting only if both are accepted. Thus, \emph{Pinned \QMA} can be verified in \emph{\QMA}.
On the other hand, for any \emph{\QMA} verifier circuit $W'$ that demands a witness $\ket{\phi'}$, there exists a pinned version, which demands a witness $\ket{\phi} = \ket{\phi'}\ket{0}$ with one extra qubit, and whose verifier circuit $W$ simply disregards the pinned qubit and verifies only the $\ket{\phi'}$ part with $W'$.

However, things are not quite as straightforward when 
instead of \emph{\QMA} witnesses we start pinning qubits of low energy states for the \emph{Local Hamiltonian problem}. Let us consider the \emph{\QMA}-complete problem 
\emph{Local Hamiltonian (LH)}, and investigate the pinning requirement. Imagine we look at a Hamiltonian $H'$, and ask if there exists a low energy state of the form $\ket{\psi'} = \ket{\psi}\ket{0}$. We call this problem \emph{Pinned LH}. 

\begin{definition}[The $p$-Pinned $k$-Local Hamiltonian Problem]
\label{def:pinLH}
Consider a $k$-local Hamiltonian $H$ for a system of size $n$, a $p$-qubit state vector $\ket{\phi}$,
with $p=\textrm{poly}(n)$
and two energy bounds $b$, $a$, such that $b-a \geq {1}/{\textrm{poly}(n)}$.
You are promised that either: 
\begin{itemize}
	\item[YES] There exists an $n-p$ qubit state vector $\ket{\psi}$, such that the energy of the $n$-qubit state vector  
	$\ket{\psi}\ket{\phi}$ with respect to $H$ is at most $a$, or
	\item[NO] for any state vector $\ket{\psi}$, the energy of the $n$-qubit state  vector $\ket{\psi}\ket{\phi}$ with respect to $H$ is at least $b$.
\end{itemize}
Decide, which is the case.
\end{definition}

We will prove the following theorem:
\begin{theorem}[\QMA-completeness of the Pinned $k$-Local Hamiltonian Problem]
The Pinned $k$-Local Hamiltonian Problem is \emph{\QMA}-complete.
\end{theorem}
\begin{proof}
First, on the one hand, \emph{Pinned LH} is no easier than \emph{LH}, because for any local Hamiltonian $H$, we can choose choose $\ket{\phi}=\ket{0}$ and set up
\begin{align}
	H' = H \otimes \ii.
\end{align}
There exists a low-energy state of $H'$ of the form $\ket{\psi}\ket{0}$ if and only if there exists a low-energy state vector $\ket{\psi}$ of $H$. Thus, \emph{Pinned LH} is 
\emph{\QMA}-hard as solving it allows one to solve the \emph{LH Problem}. 
On the other hand, observe that \emph{Pinned LH} belongs to \emph{\QMA}. We can set up a quantum verifier that receives the witness $\ket{\psi}$, adds its own single-qubit state vector $\ket{\phi}$, and then tests whether the state vector $\ket{\psi}\ket{\phi}$ has low enough energy for the \emph{Pinned Local Hamiltonian} $G'$. 
In summary, \emph{Pinned LH} is \emph{\QMA}-complete.

This could be the end of the proof. However, one might desire more details in order to understand how to translate the energy bounds between these problems. 
We can explicitly set up the \emph{LH} problem to contain \emph{Pinned LH} for example as follows.
Let us construct a Local Hamiltonian $G$, which has a low-energy state if and only if a \emph{Pinned LH} $G'$ has a low-energy state vector of the form $\ket{\psi}\ket{\phi}$. 
Without loss of generality, we can again take $\ket{\phi}=\ket{0}$, by a local basis transformation on the operators acting on the last qubit.

Let us then set up a Local Hamiltonian $G$ retaining the properties of a pinned $G'$ by penalizing the additional qubit with energy $\Delta > 0$ if 
it is not in the desired pinned state vector $\ket 0$,
\begin{align}
	G = G' + \Delta \, \ii \otimes \ket{1}\bra{1}. 
\end{align}
If there exists a state vector of the form $\ket{\psi}\ket{0}$ for the \emph{Pinned Local Hamiltonian}  $G'$ with energy $E_{\psi,0} \leq a$, the same state will also have a ``low'' energy for the local Hamiltonian $G$,
\begin{align}
	\bra{0} \bra{\psi} G \ket{\psi}\ket{0} \leq a.
	\label{abound}
\end{align}
On the other hand, if it is the case that any state vector of the form 
$\ket{\psi}\ket{0}$ has energy at least $E_{\psi,0} \geq b$, then 
taking a general state vector,
\begin{align}
	\ket{S} &= (\cos \varphi) \ket{\psi_0}\ket{0} + (\sin \varphi) \ket{\psi_1}\ket{1},
\end{align}
we can show that the ground state energy of the local Hamiltonian $G$ obeys
\begin{align}
	E_S  = \, &  \bra{S}G\ket{S} \\
	= \, & (\cos^2 \varphi) \bra{0}\bra{\psi_0}G'\ket{\psi_0}\ket{0}\nonumber
	\\&+ (\sin^2 \varphi) \left( \bra{1}\bra{\psi_1}G'\ket{\psi_1}\ket{1} + \Delta\right) \nonumber\\
	&+ (\cos \varphi \sin \varphi)
		\left( \bra{1}\bra{\psi_1}G'\ket{\psi_0}\ket{0}  + c.c.
		\right) \nonumber\\
	\geq \, & b \cos^2 \varphi + \Delta \sin^2 \varphi \\
	& + (\sin^2\varphi )\bra{1}\bra{\psi_1}G'\ket{\psi_1}\ket{1} \nonumber\\
	& + (\sin 2\varphi) \textrm{Re}\left[\bra{1}\bra{\psi_1}G'\ket{\psi_0}\ket{0}\right] \nonumber\\
	\geq \, & 
	b \cos^2 \varphi + \left(\Delta-\norm{G'}\right) \sin^2 \varphi - \sin 2\varphi \norm{G'}. 
\end{align}
Let us label $c:=\Delta-\norm{G'}$ and $d:=\norm{G'}$ to write
\begin{align}
	E_S &\geq \frac{b}{2} (1+\cos 2\varphi) + \frac{c}{2} \left( 1-\cos 2\varphi\right) - d \sin 2\varphi.
\end{align}
Assuming $c-b>d>0$, it is easy to find that the extrema of this expression appear at 
\begin{align}\tan 2\varphi = \frac{2d}{c-b},\end{align} producing
\begin{align}
	E_S &\geq 
    \frac{1}{2}\left(c+b-\sqrt{(c-b)^2+(2d)^2}\right).
\end{align}
Let us now set 
\begin{align}
c=\frac{1}{2}\left(b+a+\frac{(2d)^2}{b-a}\right),\end{align}  
i.e., 
\begin{align}
\Delta = c+d = \frac{b+a}{2}+d \left(\frac{2d}{b-a}+1\right) = \textrm{poly}(n).
\end{align} 
With basic algebra, recalling $b>a$, we can show that this satisfies
$c-\sqrt{(c-b)^2+(2d)^2}\geq a$, 
and thus
\begin{align}
	E_S &\geq \frac{a+b}{2},
\end{align}
which means in the {\em NO} instances, the ground state energy will be at least $({a+b})/{2}$, which is at least an inverse polynomial above the lower bound $a$ in the {\em YES} instances. Together with \eqref{abound}, this means we have translated the original problem's energy bounds to $a'=a$ and $b'= ({a+b})/{2}$, halving the promise gap of the original \emph{Pinned LH}.
\end{proof}

Therefore, we have not really changed the complexity of the general local Hamiltonian problem by the pinning requirement. However, the situation surprisingly changes when we start thinking about Hamiltonians whose terms come from a restricted class, as we will show in the following sections.


\subsection{Pinned Commuting Local Hamiltonian}
\label{sec:pincomm}

Pinning a qubit effectively projects into a subspace of the entire Hilbert space. 
When the original Hamiltonian comes with some restrictions, these may be lifted after this projection. 
Here and in the following sections, we investigate such cases.
First, we claim that pinning a qubit for a {\em commuting local Hamiltonian} and asking about the lowest possible energy of such a state is as difficult as asking about the ground state energy of a generic local Hamiltonian.

Note that the complexity of the original (unpinned) \emph{Commuting Local Hamiltonian} problem is an open question. The restriction to commuting terms suggests the problem is not very different from classical. Schuch has showed that this problem is in \NP\ for plaquette (4-local) interaction terms on a square lattice of qubits \cite{SchuchCommuting}, i.e. there exist classical proofs that such Hamiltonians have energy lower than some bound.
This result has been expanded and improved in work by Aharonov et al. \cite{aharonov_complexity_2011,aharonov_complexity_2018}. Importantly, though, the complexity of the problem is unknown for generic graphs, larger locality, and larger local dimension terms. Importantly, already quite simple commuting local Hamiltonians have ground states with topological order (e.g. the toric code \cite{kitaev_fault-tolerant_2003}), so the complexity of finding the properties of their ground states could be much harder. In particular, the commuting ground state connectivity problem about the structure of the ground state is \QCMA-complete. Note though, that this is likely lower than \QMA-completeness.
This all motivates us to investigate \emph{Pinned Commuting LH}. 
We will now prove our first result:

\begin{theorem}[\QMA-completeness of the Pinned Commuting $3$-local Hamiltonian problem]
\label{th:comm}
The \emph{Pinned Commuting $3$-local Hamiltonian} problem is \emph{\QMA}-complete.
\end{theorem}
The \emph{Pinned Commuting $k$-local Hamiltonian} problem is defined analogously to Definition~\ref{def:pinLH}, with an additional condition: the Hamiltonian's terms commute with each other. Let us prove it is \emph{\QMA}-complete.

\begin{proof}
First, note that the \emph{Pinned Commuting $k$-local Hamiltonian} problem is in \emph{\QMA}, just as \emph{Pinned LH} is. The harder direction is to show that commuting terms plus pinning can result in complexity equal to the case of unrestricted local Hamiltonians.
Thanks to Ref.~\cite{BiamonteLove}, we know that the $2$-local Hamiltonian problem made from 
$Z$, $X$, $ZZ$, and $XX$ terms is \emph{\QMA}-complete.
Let us take such a Hamiltonian and split it into two groups, one made from $ZZ$ and $Z$ terms, and the other made from $XX$ and $X$ terms.
The terms within each group commute with each other. 
Let $H = \sum_i A_i + \sum_j B_j$ be such a non-commuting $k$-local Hamiltonian, where in the group $A = \sum_i A_i$, all the $A_i$ commute with each other, and in $B=\sum_j B_j$, all the terms $B_j$ commute with each other. Assume the Local Hamiltonian promise problem for this $H$ has energy bounds $b$ and $a$.
Let us now add another qubit to the system, and modify the terms to
\begin{align}
	A'_i &= A_i \otimes \frac{1}{2}\left(\ii+X\right)_{n+1}
		= A_i \otimes \ket{+}\bra{+}_{n+1}, \label{commutingpin}\\
	B'_j &= B_j \otimes \frac{1}{2}\left(\ii-X\right)_{n+1}
		= B_j \otimes \ket{-}\bra{-}_{n+1}, \nonumber
\end{align}
similarly to the approach taken in Ref.~\cite{GMV}.
These terms form a fully commuting, $(k+1)$-local Hamiltonian $H' = \sum_i A'_i + \sum_j B'_j$. How much power would we have if we could figure out whether $H'$ has a low-energy state vector of the form $\ket{\psi}\ket{0}$?
Observe on the one hand that when we pin the last qubit to the state vector  $\ket{0}$, the expectation values of the $A_i$'s and $B_j$'s become
\begin{align}
	\bra{0}\bra{\psi}A'_i \ket{\psi}\ket{0} 
		&= \frac{1}{2} \bra{\psi}A_i \ket{\psi},
\\
	\bra{0}\bra{\psi}B'_j \ket{\psi}\ket{0} 
		&= \frac{1}{2} \bra{\psi}B_j \ket{\psi}.
\end{align}
Thus, if the original $k$-local $H$ has a ground state vector $\ket{\psi}$ with energy $a$, the state vector $\ket{\psi}\ket{0}$ will have energy 
${a}/{2}$ for the new commuting Hamiltonian $H'$, as $\bra{0}\bra{\psi}H'\ket{\psi}\ket{0} = \frac{1}{2}\bra{\psi}H\ket{\psi}$.
On the other hand, if the energy of any state vector $\ket{\psi}$ for the Hamiltonian $H$ is at least $b$, 
the energy of any state vector $\ket{\psi}\ket{0}$ for the new commuting Hamiltonian $H'$ is at least ${b}/{2}$.

Therefore, if one could solve a \emph{Pinned Commuting $(k+1)$-Local Hamiltonian} problem on 
$n+1$ qubits, with promise $\frac{b}{2},\frac{a}{2}$, one could use this to solve 
a \emph{$k$-Local Hamiltonian} problem (made from two commuting groups of terms) on an $n$ qubit systems, with promise bounds $b,a$. As the original problem is \emph{\QMA}-hard for $k=2$, we have thus proven that \emph{$3$-local Pinned Commuting Local Hamiltonian} is \emph{\QMA}-complete.
\end{proof}

Note that our construction is not geometrically local, as it requires interaction with the pinned qubit for all original particles. We leave the possibility of geometric locality as an open question.

\subsection{Pinned Stoquastic Local Hamiltonian}
\label{sec:pinstoq}

Let us look at another restricted class -- stoquastic Hamiltonians with non-positive off-diagonal terms. 
For such Hamiltonians an important obstacle to classical simulation via Quantum Monte Carlo -- the sign problem -- does not arise \cite{loh_sign_1990}. 
The local Hamiltonian for stoquastic Hamiltonians defines the complexity class \emph{\StoqMA}  \cite{bravyi_complexity_2006}, which is believed to be strictly smaller than \emph{\QMA} for the above reason.
In particular, stoquastic Hamiltonians are not thought to be universal for quantum computing.
What happens when we pin some of the qubits of such Hamiltonians? We show the following.

\begin{theorem}[\QMA-completeness of the Pinned Stoquastic 3-Local Hamiltonian problem]
\label{th:stoq}
The Pinned Stoquastic 3-Local Hamiltonian problem is \emph{\QMA}-complete.
\end{theorem}

A different viewpoint on this problem is given in Ref.\ \cite{JGL}, where the authors show universality of adiabatic evolution in the highest excited state of a stoquastic Hamiltonian, and the \emph{\QMA} hardness of lower bounding the highest energy of such a Hamiltonian.

\begin{proof}
As in the proof of Theorem~\ref{th:comm}, we start with observing that Pinned stoquastic $k$-local Hamiltonian is in \emph{\QMA}, because \emph{Pinned LH} is in \emph{\QMA}. We will now show that looking at the ground state energy of a Hamiltonian with stoquastic terms with pinning a qubit results is as hard as for a general local Hamiltonian.

Let us start with an instance of the \emph{\QMA}-complete problem \emph{Local Hamiltonian}.
For each such Hamiltonian $H$, we can write another using only stoquastic terms, in order to deal with possible positive off-diagonal elements in $H$. For this, we will divide $H=\hat{O} + \hat{P}$ into local terms $\hat{O}$ which are diagonal or have negative off-diagonal elements, and local terms $\hat{P}$ with positive off-diagonal elements. Let us replace the latter with stoquastic terms as follows. 
First, add an extra qubit $q$ in a state vector $\ket{-} = \left(\ket{0}-\ket{1}\right)/\sqrt 2$ to the system. Second, modify each term $\hat{P}$ by attaching 
the operator $X_q$ and change its sign, generating a new, stoquastic Hamiltonian $H' = \hat{O}\otimes \ii -\hat{P}\otimes X_q$.
When we then look at state vectors of the form $\ket{\phi}\ket{-}$, the expectation values of the modified Hamiltonian will be 
\begin{align}
	\bra{-}\bra{\phi} H' \ket{\phi}\ket{-}_q &= 
    \bra{-}\bra{\phi} \hat{O} \otimes \ii- \hat{P} \otimes X_q \ket{\phi}\ket{-}_q
    \nonumber\\
    & = \bra{\phi} \hat{O} + \hat{P} \ket{\phi} = \bra{\phi} H \ket{\phi}. \label{stoqexpect}
\end{align}
The expectation value of a pinned state vector $\ket{\phi}\ket{-}$ for the stoquastic $H'$ is the same as for the state  vector $\ket{\phi}$ and the original Hamiltonian $H$.

In more detail, let us start with the \emph{\QMA}-complete $2$-local Hamiltonian made from terms $X, Z, X\otimes X, Z\otimes Z$, and $X\otimes Z$
\cite{BiamonteLove}. First, we will change each term of the $X$ type with a positive prefactor $x_a>0$ into
\begin{align}
	x_a X_a \qquad \mapsto \qquad - x_a X_a \otimes X_q,
\end{align}
which is stoquastic. When we pin the qubit $q$ in the state vector $\ket{-}_q$, the expectation value of the new term in the state vector $\ket{\phi}\ket{-}_q$
will be simply $x_a \bra{\phi} X_a \ket{\phi}$, thanks to $\bra{-}X_q\ket{-}=-1$.
We can deal with the terms of the type $XX$ with a positive prefactor just as easily.
Next, we will look at the terms $X\otimes Z$ in $H$, whose off-diagonal terms have a varying sign.
Because we can rewrite $X\otimes Z = X \otimes \ket{0}\bra{0} - X \otimes \ket{1}\bra{1}$, 
assuming $x_{a,b}>0$, the corresponding terms in $H'$ will be
\begin{align}
	x_{a,b}  X_a & \otimes Z_b  \label{stoqpin}\\&\mapsto  - x_{a,b} X_a \otimes \left( \ket{0}\bra{0}_b \otimes \,X_q 
 + \ket{1}\bra{1} \otimes \,\ii_q\right),\nonumber\\
-x_{a,b}  X_a & \otimes Z_b  \\ & \mapsto  
- x_{a,b} X_a \otimes \left( \ket{0}\bra{0}_b \otimes \,\ii_q 
+ \ket{1}\bra{1} \otimes \,X_q\right).                                     
                                    \nonumber
\end{align}
Observe that the modified terms are stoquastic, with only negative off-diagonal elements.

Consider now the new stoquastic $3$-local Hamiltonian $H'$ and ask whether its low-energy vectors can have the form $\ket{\phi}\ket{-}$. On the one hand, if the original $H$ has a ground state vector $\ket{\phi}$ with energy $a$, the state vector
$\ket{\phi}\ket{-}$ will have energy $a$ for the new stoquastic Hamiltonian $H'$.
On the other hand, if the energy of any state vector $\ket{\phi}$ for the Hamiltonian $H$ is at least $b$, 
the energy of any state vector of the form $\ket{\phi}\ket{-}$ is at least $b$ for the new commuting Hamiltonian $H'$.
Therefore, we have turned a \emph{Local Hamiltonian} problem with promise parameters $a,b$, into a \emph{Pinned Stoquastic Local Hamiltonian} with the same promise, with a doubled Hilbert space (adding a qubit), and stoquastic terms that have a locality increased by 1.
Solving \emph{Pinned Stoquastic LH} is thus at least as hard as \emph{LH}, and thus \emph{\QMA}-complete.
\end{proof}

Note that in the proof we provided, the type of the terms in $H'$ is different from $H$, as we were only interested in making them stoquastic, not keeping their form. It remains open to analyze what is the hardness of \emph{Pinned Stoquastic Hamiltonian} with restricted form (e.g., only XXX, ZZZ) or locality below 3.
After showcasing the pinning technique in two examples, we will continue exploring how far it takes us, applying it to simpler and simpler original Hamiltonians.


\subsection{Pinned Permutation Hamiltonians}
\label{sec:pinperm}

The possibilities opened in the previous sections motivate us to go further and design a classically looking problem about $0/1$ permutation matrices that will still be \emph{\QMA}-complete. This is a further restriction on stoquastic Hamiltonians.
We claim the following.
\begin{theorem}[\QMA-completeness of the Pinned Local Permutation Hamiltonian]
\label{th:perm}
Pinned Local Permutation Hamiltonian is \emph{\QMA}-complete, 
with a logarithmic number of pinned qubits.
\end{theorem}
Note that (dynamical) universality for quantum computation with $0/1$ matrices has been previously demonstrated for example in the \PromiseBQP\string-rewriting problem of Wocjan and Janzing \cite{JanWocRewrite}, or the universal computation by quantum walk construction of Childs et al. \cite{ChildsWalk}.

\begin{proof}
One direction of Theorem~\ref{th:perm} is easy -- pinned local permutation Hamiltonian is obviously in \emph{\QMA}.
The more difficult part is again to construct \emph{\QMA}-hard instances of pinned $0/1$ Hamiltonian.
First, we will take a target Hamiltonian made from Pauli matrices, and replace them by $0/1$  matrices on a larger Hilbert space, with a technique similar to
those of Ref.~\cite{JGL}, where it has been used to build \emph{\QMA}-hard instances of stochastic matrices. Second, we will utilize pinning to generate the desired real-valued prefactors for the permutation, and thus also the effective original Pauli terms.

Consider an instance of the \emph{\QMA}-complete, $2$-local Hamiltonian problem with a Hamiltonian $H$ made from $X$, $Z$, $XX$ and $ZZ$ terms, as in Section~\ref{sec:pincomm}, with real-valued prefactors. Let us deal with Pauli terms first, and consider the prefactors later. The $X$ and $XX$ terms already are permutation matrices. For the $Z$ and $ZZ$ terms, we will add an auxiliary qubit $z$, and transform the interactions as
\begin{align}
	Z \quad \mapsto  \quad& \ket{0}\bra{0} \otimes \ii_z + \ket{1}\bra{1} \otimes X_z, \\
    Z\otimes Z \quad \mapsto \quad&  \left(\ket{0,0}\bra{0,0}+\ket{1,1}\bra{1,1} \right)\otimes \ii_z \\& + \left(\ket{0,1}\bra{0,1}+\ket{1,0}\bra{1,0} \right) \otimes X_z,\nonumber
\end{align}
generating $2$-local and $3$-local permutation matrices, made from $0/1$  elements. 
This results in a permutation Hamiltonian $H'$. When we pin the aixuliary qubit $z$ in the state vector $\ket{-}$, we can effectively generate the original $Z$ and $ZZ$ (and of course $X$ and $XX$) terms as we did for stoquastic Hamiltonians.

Second, we want to generate real-valued prefactors for the effective Pauli terms using permutation Hamiltonians. This is straightforward with the help of pinning, once we add and pin several auxiliary systems.
In the definition of \emph{Pinned Local Hamiltonian}, we allow for pinning of up to a polynomial number of qubits. 
In the problems considered so far, we pinned a single qubit. Here, we will use a logarithmic number of such auxiliary systems.

Let us start with a system described by the Hamiltonian $H'=\sum_i P_i$ built in the previous step as a sum of permutation matrices, with only $0,1$ elements, with a single $1$ in each row and column.
We will show how to add $Q+1$ qubits and interactions to form $H''$.
Pinning the $Q$ new auxiliary systems to a specific product subspace $\mathcal{S}$,
will then allow us to effectively investigate the target Hamiltonian $H=\Pi_{\mathcal{S}} H'' \Pi_{\mathcal{S}}$ with the desired form, up to precision $2^{-Q}$ for its terms.
This precision comes from the possibility of imprecisions of the original Local Hamiltonian problem.
If the original problem was given precisely, but with an inverse polynomial promise gap,
allowing for an inverse-polynomial imprecision in the Hamiltonian's elements simply shrinks the promise gap, if we consider a large enough $Q$, which is however still logarithmic in $n$.

Recall our target effective Hamiltonian $H$ has general real prefactors for its Pauli terms.
Let us consider the terms from the permutation Hamiltonian $H'$ from the first step.
Imagine we want the term $\hat{O}$ to have a prefactor $0<x<1$.
We will decompose $x$ into binary, up to some precision $Q$, as 
\begin{align}
x=\sum_{j=1}^{Q} \frac{x_j}{2^j},\end{align} with $x_j \in \{0,1\}$.
For each nonzero $x_j$, we will pin an auxiliary qubit $q_j$ to 
\begin{align}
	\ket{\alpha_j} = \cos \alpha_j \ket{0} + \sin \alpha_j \ket{1},\qquad \sin 2 \alpha_j = \frac{1}{2^j}, \label{permpin}
\end{align}
with a new term $\hat{O} \otimes X_{q_j}$ in $H'$ for each nonzero $x_j$,
in order that $\bra{\psi}\bra{\alpha_j} \hat{O}\otimes X_{q_j}\ket{\psi}\ket{\alpha_j}
=\bra{\psi}\hat{O}\ket{\psi}/2^j$.
Pinning the $Q$ auxiliary qubits to their respective state vectors $\ket{\alpha_j}$, 
altogether they become an effective Hamiltonian 
\begin{align}\left(\sum_{j=1}^{Q} \frac{x_j}{2^j}\right)\hat{O}
\end{align} on the $n$ qubits of the system.
Second, to generate effective negative prefactors, we use the standard trick from before, adding the auxiliary qubit $q_0$ pinned in the state vector $\ket{-}$, and an interaction of the form $\hat{O}\otimes X_{q_0}$ to the desired terms.

Let us summarize. 
Our target Hamiltonian $H$ acting on $n$ qubits has $M$ Pauli terms with real prefactors and locality at most 2. 
In step 1, we built an $n+1$ qubit permutation Hamiltonian $H'$ with locality at most 3,
which did not yet include the desired real prefactors.
In step 2, we constructed the final permutation Hamiltonian $H''$ which works on $n+Q+2$ qubits, and has at most $2M\times (Q+1)$ terms, with locality at most 5.
We pinned the auxiliary qubits $z$ and $q_0$ into the state vector $\ket{-}$, and the auxiliary qubits $q_1,\dots,q_Q$ into the states \eqref{permpin}. Determining the lowest energy of the pinned 5-local permutation Hamiltonian $H''$, with $Q+2$ pinned qubits, is thus \emph{\QMA}-hard, as it implies determining the ground state energy of the target local Hamiltonian $H$, an instance of the \emph{\QMA}-complete problem Local Hamiltonian.
Therefore, the Pinned permutation Hamiltonian problem is \emph{\QMA}-hard, as well as \emph{\QMA}-complete.
\end{proof}


\section{A dynamical view of pinning}
\label{sec:dynpin}

In the previous section, we looked at how pinning can contribute to the complexity of determining the static properties of local Hamiltonians -- the bounds on the energies of states from the pinned subspace. We now turn to a dynamical question, asking what pinning can contribute when applied to an evolution with a local, time-independent Hamiltonian with a restricted set of interactions (or unitaries). We will consider constantly measuring one qubit in a particular basis, pinning it via the {\em Zeno effect} to a particular state.
Note that this is different from {\em postselection.} There, one is allowed to choose a particular result of a measurement of a subsystem without regard of the result's (im)probability. This would give one immense computational power \cite{postselection}, as postselected quantum computation has the power of \PP, much larger than \NP.
Pinning does not allow us to choose a measurement result freely. Instead, we must rely on the Zeno effect to give us a high probability of the desired projection. Pinning is thus applicable in practice, unlike the theoretical concept of postselection. 

With frequent projective measurement, we effectively get access to a specific state of a qubit, and thus a specific subspace of the whole Hilbert space. We will show that the dynamics of restricted Hamiltonians in this chosen subspace can result in universal dynamics. In the circuit model, we know that access to specific states can greatly enhance the power of a restricted model. For example, a source of magic states is enough to turn computation with Clifford gates
into a universal quantum computation \cite{MagicStatesBravyi}.
Following a similar strategy as in the previous section, we will now show how to get universal quantum computation out of evolution with a restricted set of (e.g. commuting) Hamiltonians together with a fixed Pauli basis measurement of a \emph{single} qubit.


\subsection{Warm-up: evolution with pinned stoquastic Hamiltonians}
\label{sec:pinstoqevol}

We will start with a simple example of applying pinned evolution to stoquastic Hamiltonians. We know that evolution with stoquastic Hamiltonian is already universal for quantum computation, as shown by Childs et al. \cite{childs_universal_2009}. However, our pinned construction has its own merits, even over later developments \cite{ChildsWalk}, in terms of space/time requirements.
Moreover, it will be useful in Section~\ref{sec:stoqGSCON}, where we will use it for the proof of \QCMA-hardness of the \emph{Stoquastic GSCON} problem.
Note also that Fujii has shown how adding local measurements to adiabatic evolution with stoquastic Hamiltonians (stoqAQC) results in universality (for adaptive measurements) or quantum advantage (if using non-adaptive measurements) \cite{fujii_quantum_2018}. Again, what we do here is more efficient, requires smaller locality and easier control.

Let us then look at a system with a stoquastic Hamiltonian
\begin{align}
	H' = A \otimes \ii_{q} + B \otimes X_{q},	
\end{align}
made from two groups of local, stoquastic terms $A$ and $B$, with no positive off-diagonal entries.
Furthermore, we demand $B$ to be entirely off-diagonal.
The terms $B\otimes X_q$ include an interaction with an auxiliary qubit $q$, similarly to \eqref{stoqpin}. 
We can now show that pinned evolution with time-independent, local stoquastic Hamiltonians is universal for \BQP\ as follows.

We initialize the auxiliary qubit $q$ as $\ket{-}$, and measure it in the $X$ basis often enough. This likely pins the auxiliary qubit to the state 
vector $\ket{-}$. Meanwhile, the system evolves with $H'$. This results in a particular effective evolution.
Let us cut the time evolution into small steps of size $\delta \rightarrow 0$.
The evolution can be approximated as alternating the evolution $e^{-i\delta H}$ with a projection of the last qubit onto the state vector $\ket{-}$. It will be helpful to express
\begin{align}
	\bra{-} e^{-i\delta H}&  \ket{\psi}\ket{-}_q
		\approx \bra{-} e^{-i\delta A} e^{-i\delta B \otimes X_q} \ket{\psi}\ket{-}_q \nonumber\\
		&\approx \bra{-} (1 - i\delta A - i \delta B\otimes X_q) \ket{\psi}\ket{-}_q \nonumber\\
        &=  1 - i\delta (A-B)\ket{\psi} \approx e^{-i\delta (A-B)}\ket{\psi},
\end{align}
valid up to first order in $\delta$. This allows us to effectively evolve the state vector $\ket{\psi}$ with the general, non-stoquastic Hamiltonian $H=A-B$.

Moreover, because the last qubit is in an eigenstate of $X_q$, it never gets flipped into the state vector $\ket{+}$.
Thus, taking $\delta = {t}/{N}$, with $N \rightarrow \infty$, we can confidently say that
\begin{align}
	\ket{\psi(t)}_{\textrm{PE}} & = \left(\prod_{j=1}^{N} P_{-} \left(e^{-i\delta (A+B\otimes X)}\right)\right) \ket{\psi}
    \nonumber\\
   &  = e^{-it (A-B)} \ket{\psi} + \ket{\delta},
\end{align}
where $\ket{\delta}$ is an error state vector with norm of order at most $\delta = {t}/{N}$, i.e., going to zero as $N\rightarrow\infty$.
Therefore, we can simulate evolution with a time-independent, non-stoquastic local Hamiltonian using evolution according to a stoquastic local Hamiltonian (with locality increased by 1) and pinning.
This is universal for quantum computation (\BQP), when we recall various standard constructions for universal quantum computation by evolution with a time-independent, non-stoquastic local Hamiltonian, e.g., Ref.~\cite{universal2body}.


\subsection{Pinned evolution with commuting Hamiltonians}
\label{sec:pincommevol}

Let us now turn to our main result about pinned dynamics. We will investigate what kind of evolution we can achieve with pinned commuting local Hamiltonians. A similar question has been posed in the context of Hamiltonian purification, with an emphasis on obtaining an universal algebra \cite{burgarth}. Here, we will show how to efficiently simulate evolution with a non-commuting Hamiltonian $H=A+B$, made from two groups of terms that commute within the group. For this, we will construct a Hamiltonian
\begin{align}
	H' = 2 A \otimes \ket{+}\bra{+}_{q} + 2 B \otimes \ket{-}\bra{-}_{q}
\end{align}
 all of whose terms commute, by adding an auxiliary qubit $q$ as in \eqref{commutingpin}.
Let us now analyze what happens when we alternate computational basis measurements on the last qubit, initialized as $\ket{0}$, with evolution according to $H'$.
We will prove the following.
\begin{theorem}[Universality of commuting Pinned Evolution]
\label{th:pincommevol}
Pinned evolution with time-independent, local commuting Hamiltonians is universal for \BQP.
\end{theorem}
\begin{proof}
Let us look at a short time interval $\delta = {t}/{N}$, with $N\rightarrow \infty$. The pinned evolution of the system will be well approximated by the evolution $e^{-i\delta H'}$ according to $H'$ for time $\delta$, and then a measurement in the computational basis.
This repeated measurement should on the one hand effectively pin the auxiliary qubit $q$ in the state vector $\ket{0}$,
as in Vaidman's bomb-testing procedure \cite{elitzur1993quantum} in its circuit setting \cite{kwiat1995interaction}. This is the Zeno effect, explained in detail e.g., 
in Ref.~\cite{SattathNagajQuantumMoney}, where we also find that the probability of a ``bad'' projection (a flip of the $\ket{0}$ to $\ket{1}$ scales as $O\left(\delta^2\right)$, and can be made arbitrarily small even after $O\left(\delta^{-1}\right)$ repetitions. 
On the other hand, what is the effective evolution of the rest of the system? Let us calculate
\begin{align}
	e& ^{-i\delta H'} \ket{\psi}\ket{0}
		= e^{-i2 \delta A \otimes \ket{+}\bra{+}} e^{-i2 \delta B \otimes \ket{-}\bra{-}} \ket{\psi}\ket{0} 
        \label{zenocomm1} \\
        &\approx \left(\ii - i 2 \delta A \otimes \ket{+}\bra{+}
        	 - i 2 \delta B \otimes \ket{-}\bra{-}\right)  \ket{\psi}\ket{0}  \\
        &\approx
        	\ket{\psi}\ket{0} - \frac{i 2 \delta}{\sqrt{2}} \left(
            	A \ket{\psi} \ket{+}
        	+ B \ket{\psi} \ket{-} \right) \\
       & = \left(\ii -i\delta (A+B)\right)\ket{\psi}\ket{0} - i \delta (A-B)\ket{\psi}\ket{1} \\
       & \approx e^{- i\delta (A+B)} \ket{\psi}\ket{0} - i\delta(A-B)\ket{\psi}\ket{1}, \label{zenocommfinal}
\end{align}
correct up to order $\delta$. Therefore, when we now measure the auxiliary qubit, we will get the result $0$ and obtain the state
$e^{-i\delta(A+B)}\ket{\psi}$, with probability $1-O(\left(\delta\norm{A-B}\right)^2)$.
Moreover, the state can also contain an error vector with norm $O((\delta \norm{A+B})^2)$,
as the evolution \eqref{zenocomm1} with commuting terms cannot produce mixed terms such as $AB$, while \eqref{zenocommfinal} does include them in its series expansion.

What happens when we repeat this evolve-measure procedure $N={t}/{\delta}$ times? We end up with the state vector
$e^{-it (A+B)}\ket{\psi}$, with an error vector of norm $O(t\delta \norm{A+B}^2)$, while the probability that all the $N$ measurements of the pinned auxiliary qubit result in $\ket{0}$ is lower bounded by $1-O(t\delta \norm{A-B}^2)$.
Therefore, we can simulate evolution with unrestricted (non-commuting) Hamiltonians using commuting Hamiltonians and pinning.
Starting with a universal local Hamiltonian built from two groups of commuting terms as in Section~\ref{sec:pincomm}, this directly translates into the statement of the theorem:
pinned evolution with commuting local Hamiltonians is universal for quantum computation.
\end{proof}

State preparation with a universal, 2-local, non-commuting construction \cite{universal2body} that has $O(L)$ gates in a circuit can thus be efficiently simulated with low error by time $O(L)$ evolution with 3-local, commuting terms, and frequent measurement of a single qubit.


\section{Ground state connectivity}
\label{sec:GSCON}

Our original motivation for exploring pinning was to understand better the variants of Gharibian and Sikora's \emph{Ground State Connectivity (GSCON)} problem \cite{GharibianSikora}. It asks about the possibility of traversing the low-energy subspace of a local Hamiltonian from one specific ground state to another, using local unitary transformations. 
Gosset, Mehta and Vidick \cite{GMV} have shown that the problem remains \emph{QCMA} complete even if only commuting Hamiltonians are used. 
In their proof, they use a trick similar to pinning -- combining the original Hamiltonian's terms with projections on auxiliary qubits to make the terms commute. 
Then they demand that the initial and final ground state have a few qubits in a specific state -- which means that the original non-commuting Hamiltonian's terms are effectively applied. 
Moreover, this has to be combined with the impossibility of a simple flip of this state without a computation being verified first. 
Nevertheless, it helped us realize that the \emph{GSCON}  formulation allows one to essentially fix some part of the ground state, adding extra power to restricted forms of Hamiltonians. 

Therefore, using techniques similar to Ref.~\cite{GMV}, hardness results for pinned local Hamiltonians should be translatable to hardness of \emph{GSCON}  for similarly restricted Hamiltonians. For example, we will be able to show \QCMA-hardness of \emph{GSCON}  for stoquastic Hamiltonians, building on Ref.~\cite{GMV} and the construction from Section~\ref{sec:pinstoqevol}. 
Moreover, in this context we will also provide some evidence into the free-fermionic variant of \emph{GSCON}, to be further developed in future work.


\subsection{Stoquastic GSCON}
\label{sec:stoqGSCON}

First, we will show how to build on the proof that the \emph{Ground State Connectivity (GSCON)} 
problem is \QCMA-complete for commuting Hamiltonians, as well as on universality of pinned stoquastic LH, and prove that
\emph{Stoquastic GSCON}  is \QCMA-complete.
The statement of the problem is identical to the \emph{Commuting GSCON}  problem in Ref.~\cite{GMV},  the only difference being the replacement of the word ``commuting'' by ``stoquastic''. 
We thus have: 
\begin{definition}[Stoquastic Ground State Connectivity $({H},{k},\eta_1,\eta_2,\eta_3,\eta_4,{\Delta},{l},{m},{U_\psi},{U_\phi})$]\label{def:GSCONstoq}
~

{\bf Input:}
    \begin{enumerate}
        \item {$k$}-local Hamiltonian ${H}=\sum_i H_i$ with stoquastic terms (i.e. with no positive off-diagonal elements), satisfying $\|{H_i}\|\leq 1$.

        \item $\eta_1,\eta_2,\eta_3,\eta_4, {\Delta}\in\R$, and integer ${m}\geq0$, such that 
        $\eta_2-\eta_1\geq {\Delta}$ and $\eta_4-\eta_3\geq {\Delta}$.

        \item Polynomial size quantum circuits {$U_\psi$} and {$U_\phi$} generating ``starting'' and ``target'' state vectors $\ket{\psi}$ and $\ket{\phi}$ starting from the $\ket{0}^{\otimes n}$ state, respectively, satisfying 
        $\bra{\psi}{H}\ket{\psi}\leq \eta_1$ and $\bra{\phi}{H}\ket{\phi}\leq \eta_1$.
    \end{enumerate}
 {\bf Output}:
\begin{enumerate}
    \item If there exists a sequence of {$l$}-local unitaries $(U_{i})_{i=1}^m \in U$
    such that
    \begin{enumerate}
        \item (Intermediate states remain in low energy space) For all $i\in [{m}]$ and intermediate states \\
				${\ket{\psi_i}:=U_i\cdots U_2U_1\ket{\psi}}$, one has $\bra{\psi_i}{H}\ket{\psi_i}\leq \eta_1$, and
        \item (Final state close to target state) $\norm{ U_{{m}} \cdots U_1 \ket{\psi}-\ket{\phi}} \leq \eta_3$,
    \end{enumerate}
    then output YES.
    \item If for all ${l}$-local sequences of unitaries $(U_{i})_{i=1}^{{m}}$, either:
    \begin{enumerate}
        \item (Intermediate state obtains high energy) There exists $i\in [{m}]$ and an intermediate state 
        vector ${\ket{\psi_i}:=U_i\cdots U_2U_1\ket{\psi}}$, such that $\bra{\psi_i}{H}\ket{\psi_i}\geq \eta_2$, or
        \item (Final state far from target state) $\norm{ U_{{m}} \cdots U_1 \ket{\psi}-\ket{\phi}} \geq \eta_4$,
    \end{enumerate}
    then output NO.
\end{enumerate}
\end{definition}
There is not that much that we need change in the proof of Theorem 6 in Ref.~\cite{GMV}, when we want to build a generic effective Hamiltonian from stoquastic instead of commuting terms, using ``pinning'' thanks to a restriction on the initial and final states, as well as the form of the Hamiltonian that we construct.

\begin{theorem}[\QCMA-completeness of the Stoquastic Ground State Connectivity Problem]
The Stoquastic Ground State Connectivity Problem is \QCMA-complete.
\end{theorem}
\begin{proof}
It is straightforward to see that the \emph{Stoquastic GSCON}  is in QCMA, with a witness encoding the sequence of unitaries, verifiable by a quantum computation. 
For the other direction, we are directly inspired by the proof of \QCMA-completeness of \emph{Commuting GSCON}  \cite{GMV}. 
There, the authors split a target generic (non-commuting) local Hamiltonian $G=A+B$ into two groups of local commuting terms, add two 3-qubit auxiliary registers, and set up the commuting Hamiltonian
\begin{align}
	A \otimes \Pi_{\mathcal{S}} \otimes \Pi_{+} + B \otimes \Pi_{\mathcal{S}} \otimes \Pi_{-} + \ii \otimes \ii \otimes \Pi_{\mathcal{S}}, \label{commgscon}
\end{align}
where $\Pi_{\mathcal{S}}$ projects onto $\mathcal{S} =\mathrm{span}\left\{\ket{0,0,0},\ket{1,1,1}\right\}$,
and $\Pi_{\pm}$ are projectors onto $\left(\ket{0,0,0}\pm \ket{1,1,1}\right)/\sqrt{2}$.
The \QCMA-hard \emph{GSCON}  question concerns the possible low-energy traversal from the state vector $\ket{0}^{\otimes n}\ket{1}^{\otimes 3}\ket{0}^{\otimes 3}$ to the state vector $\ket{0}^{\otimes n}\ket{0}^{\otimes 3}\ket{0}^{\otimes 3}$ by $2$-local operations. This is possible by using the first $n$-qubit register to prepare a low-energy witness for the Hamiltonian $G=A+B$. This effectively ``turns off'' the first two terms in \eqref{commgscon}, allowing one to flip the middle register to $\ket{1,1,1}$ by $2$-local operations without a high energy cost. Finally, one uncomputes the first register. Meanwhile, the last register stays ``pinned'' in $\ket{0,0,0}$, making sure both groups of terms $A$ and $B$ are in play and contribute significantly to the energy of the intermediate states. For more details, 
see the proof of Theorem 6 of Ref.~\cite{GMV}.

Let us then work out the stoquastic version of this. We start with an $n$-qubit register, and the target generic, non-stoquastic, $2$-local, $n$-qubit Hamiltonian $H$ made from $ZZ$, $ZX$, $XX$, $Z$ and $X$ terms. The Local Hamiltonian problem for this variant of $H$ is \emph{\QMA}-complete. The \emph{GSCON}  problem based on $H$ is thus \QCMA-complete.

We will construct a stoquastic \emph{GSCON}  Hamiltonian $H''$ similarly to \eqref{commgscon}, with a few important differences. 
First, let us define two operators 
\begin{align}
	Q = \frac{1}{3} \left( X_{q_1} + X_{q_2} + X_{q_3}\right),
\end{align}
an analogue of $X_q$ from Section~\ref{sec:pinstoq}, effectively flipping the sign when the auxiliary register is in the state vector $\ket{-}^{\otimes 3}$, and
\begin{align}
    R_{3} =  \frac{3}{4}\,\ii - \frac{1}{4}\left( X_{q_1}X_{q_2} + X_{q_2}X_{q_2} + X_{q_1}X_{q_3} \right),
    \label{R3}
\end{align}
a $2$-local, stoquastic operator equivalent to the projector onto the space orthogonal to the span of $\ket{-}^{\otimes 3}$ and $\ket{+}^{\otimes 3}$. 

Second, let us add a 3-qubit auxiliary register and combine the original Hamiltonian $H$ with the operator $R_3$ as $H' = H\otimes R_3$. 
Similarly to Section~\ref{sec:pinstoq}, we can split this local Hamiltonian $H'$ acting on $n+3$ qubits into groups of local terms $H' = \hat{O}' + \hat{P}'$, with non-positive off-diagonal terms $\hat{O}'$ and a group of strictly off-diagonal local terms with positive elements $\hat{P}'$.

Finally, we combine the group $\hat{P}'$ with the operator $Q$ on the final auxiliary register, in order to ensure that $-\hat{P}'\otimes Q$ is stoquastic, with strictly negative off-diagonal elements, as $\hat{P}'\otimes Q$ is a tensor product of two operators which each have strictly positive off-diagonal elements and no diagonal elements. Altogether, we arrive at the local, stoquastic Hamiltonian
\begin{align}
	H'' = \hat{O}' \otimes \ii - \hat{P}' \otimes Q + \ii \otimes R_{3}.
\end{align}
  Observe that for the state vectors of the form $\ket{\psi}\ket{-}^{\otimes 3}\ket{-}^{\otimes 3}$ and
  $\ket{\psi}\ket{+}^{\otimes 3}\ket{-}^{\otimes 3}$, the expectation value of $H''$ is zero. Meanwhile,
  when the middle register is in an $X$-basis state vector $\ket{x_1, x_2 , x_3}$ other than $\ket{-}^{\otimes 3}$ or $\ket{+}^{\otimes 3}$, and the last register remains 
  in $\ket{-}^{\otimes 3}$, the expectation value 
  \begin{align}
  	\bra{\psi}& \bra{x_1, x_2 , x_3}\bra{-}^{\otimes 3} H'' \ket{\psi}\ket{x_1, x_2 , x_3}\ket{-}^{\otimes 3}\nonumber\\
    &  =\bra{\psi}\bra{x_1, x_2 , x_3} \hat{O}'+ \hat{P'} \ket{\psi}\ket{x_1, x_2 , x_3} \\
    & =\bra{\psi}\bra{x_1, x_2 , x_3} H \otimes R_3 \ket{\psi}\ket{x_1, x_2 , x_3} 
    = \bra{\psi}H\ket{\psi} 
    \label{stoqGSCONexpect}
  \end{align}
  is equivalent to the expectation value of the original non-stoquastic Hamiltonian $H$ acting on $\ket{\psi}$, thanks to 
  \begin{equation}
  \bra{x_1, x_2 , x_3}R_3\ket{x_1, x_2 , x_3} = 1.
  \end{equation}
 The hard ground space traversal question we 
 ask is then: 
 Decide, if starting in the state vector $\ket{0}^{\otimes n}\ket{-}^{\otimes 3}\ket{-}^{\otimes 3}$, 
 one can traverse the low-energy subspace of $H''$ without energy above $\alpha$ (where this bound comes from the \QCMA-complete \emph{LH} problem with energy bounds $\alpha$ and $\beta$)
 and at most $\eta_3$ far from the state
$\ket{0}^{\otimes n}\ket{+}^{\otimes 3}\ket{-}^{\otimes 3}$, using a sequence of $2$-local unitaries of length polynomial in $n$,
or whether one must end at least $\eta_4$ far from the final state, or some of the intermediate states have energy at least $\eta_2$?

Showing completeness is straightforward with the following sequence of transformations.
Note the third register stays in $\ket{-}^{\otimes 3}$ throughout the process.
First, we prepare the low-energy witness for $H$ in the first register. The energy is zero during this process. 
Second, we flip the second register from $\ket{-}^{\otimes 3}$ to $\ket{+}^{\otimes 3}$, qubit by qubit. In this process, the energy of the states is at most $\alpha$, thanks to \eqref{stoqGSCONexpect}. Finally, we uncompute the first register, keeping the energy zero. 

For soundness, one can directly follow \cite{GMV} to show that no sequence of $2$-local unitaries will satisfy well enough the two conditions -- end near enough the final state and stay low enough in energy throughout the sequence. 
The lower bound on the energy of the intermediate states if one is to end up close to the final state is in this case $\eta_2 = \Omega\left(\beta^2/m^6\right)$, just as in the proof of Soundness of Theorem 10 in Ref.~\cite{GMV}, where $\bra{\psi}H\ket{\psi} \geq \beta$ is the bound in the NO case of the original \emph{LH} problem and $m$ is the number of unitaries in the sequence.
One has only to replace
\begin{align}
	P_0 = \ket{0,0,0}\bra{0,0,0} \quad &\mapsto \quad \ket{-}\bra{-}^{\otimes 3}, \\
    P_1 = \ket{1,1,1}\bra{1,1,1} \quad &\mapsto \quad \ket{+}\bra{+}^{\otimes 3},
\end{align}
and follow the proof.
\end{proof}

Observe that in the {\em NO} case, to obtain soundness, an efficient (poly-length) sequence of $2$-local transformations keeping the energy of intermediate states low enough simply could not exist,
and this was guaranteed by the lower bound from the Small Projection Lemma~8 \cite{GMV}.
Would this be also true in other settings besides history state preparation connected to \emph{QCMA}-complete problems? We ask this question about quantum memories, e.g., based on the toric code, in forthcoming work.

\subsection{Ground state connectivity for free fermions}
\label{sec:dynamicalmatch}

In the context of studies of Majorana fermionic quantum memories, 
variants of \emph{GSCON} for free fermions are particularly interesting
\cite{Superconducting,Litinski2017}. 
Here we provide insights that we expect to be helpful in tackling this version of the problem relevant when 
assessing Majorana fermionic
quantum memories: 
we provide evidence that between any pair of low-energy free-fermionic states, there exists a local free-fermionic circuit that interpolates between them within the low-energy subspace. 
Before we get there, let us define the Free Fermionic Ground State Connectivity Problem, though.
Note also that our discussion of the free-fermionic problem does not rely on pinning, but complements our understanding of \emph{GSCON} in a practically relevant setting.

 \begin{definition}[Free Fermionic Ground State Connectivity $({H},{k},\eta_1,\eta_2,\eta_3,\eta_4,{\Delta},{l},{m},{U_\psi},{U_\phi})$]\label{def:GSCON}
 ~
 \begin{enumerate}
 \item {{\rm Input parameters:}}
     \begin{enumerate}
         \item {$k$}-local free fermionic Hamiltonian ${H}=\sum_i H_i$ acting on $n$ fermionic modes with each $H_i $ 
         being supported on no more than $k$ modes, satisfying $\|{H_i}\|\leq 1$.
         \item $\eta_1,\eta_2,\eta_3,\eta_4, {\Delta}\in\R$, and integer ${m}\geq0$, such that f
         $\eta_2-\eta_1\geq {\Delta}$ and $\eta_4-\eta_3\geq {\Delta}$.

         \item Polynomial size fermionic Gaussian 
         quantum circuits {$U_\psi$} and {$U_\phi$} generating ``starting'' and ``target'' fermionic Gaussian 
         state vectors $\ket{\psi}$ and $\ket{\phi}$ (starting from the fermionic vacuum), respectively, satisfying 
         $\bra{\psi}{H}\ket{\psi}\leq \eta_1$ and $\bra{\phi}{H}\ket{\phi}\leq \eta_1$.
     \end{enumerate}
    
 \item {\rm Output}:
 \begin{enumerate}
     \item If there exists a sequence of {$l$}-local unitaries $(U_{i})_{i=1}^m \in U$ supported on $m$ modes each
    such that
    \begin{enumerate}
         \item (Intermediate states remain in low energy space) For all $i\in [{m}]$ and intermediate states \\
 				+
 				${\ket{\psi_i}:=U_i\cdots U_2U_1\ket{\psi}}$, one has $\bra{\psi_i}{H}\ket{\psi_i}\leq \eta_1$, and
         \item (Final state close to target state) $\norm{ U_{{m}} \cdots U_1 \ket{\psi}-\ket{\phi}} \leq \eta_3$,
     \end{enumerate}
     then output YES.
     \item If for all ${l}$-local sequences of unitaries $(U_{i})_{i=1}^{{m}}$, either:
     \begin{enumerate}
         \item (Intermediate state obtains high energy) There exists $i\in [{m}]$ and an intermediate state 
         vector ${\ket{\psi_i}:=U_i\cdots U_2U_1\ket{\psi}}$, such that $\bra{\psi_i}{H}\ket{\psi_i}\geq \eta_2$, or
        \item (Final state far from target state) $\norm{ U_{{m}} \cdots U_1 \ket{\psi}-\ket{\phi}} \geq \eta_4$,
     \end{enumerate}
     then output NO.
 \end{enumerate}
 \end{enumerate}
 \end{definition}

Here, we do not assess the hardness of the \emph{Free Fermionic GSCON} problem. 
We conjecture that in contrast to the general case, in free fermions there will always exist a local low-energy path between any pair of low-energy quantum states.
 
\begin{conjecture}[Free Fermionic Ground State Connectivity] 
For any free fermionic Hamiltonian $H$ and any pair of low-energy Gaussian fermionic states $\ket \psi, \ket \phi$ there exists a $2$-local finite Gaussian fermionic circuit interpolating between them such that all intermediate states satisfy the energy constraint. 
 \end{conjecture}

We here provide evidence in favour of this conjecture. 
 Let us denote the fermionic covariance matrix
 of the initial state vector $|\psi\rangle$
 with $\gamma$ (in the conventions of 
 Ref.\ \cite{PhysRevB.97.165123}), and with
 $\omega$ the covariance matrix of the final
 state vector $|\phi\rangle$.
 For $n$ modes, this is 
 a real $2n\times 2n$
 matrix satisfying
 $\gamma=-\gamma^T$ (as is the case for any
 covariance matrix)
 and $\gamma^T\gamma=
 \ii$ (reflecting purity). The application of Gaussian fermionic
 gates to achieve
 $|\psi_i\rangle=U_i\cdots
 U_2 U_1|\psi\rangle$
corresponds to a transformation
 \begin{equation}
    \gamma_i:= O_i \cdots O_2
     O_1 
     \gamma O_1^T
     O_2^T \cdots 
     O_i^T
 \end{equation}
 with $O_i \in SO(2n)$ for all $i$,
on the level of covariance matrices. 
In the Free Fermionic Ground State Connecttivity Problem,  
the initial covariance matrix can be written as 
\begin{equation}   
\gamma = O \gamma_0 O^T
 \end{equation}
 with $O\in SO(2n)$ and either
\begin{equation}    
\gamma_0= \bigoplus_{j=1}^n \left[
\begin{array}{cc}
0 & 1\\
-1 & 0
\end{array}\right]
\end{equation}
or  
\begin{equation}    
\gamma_0= \left(\left[\begin{array}{cc}
0 & -1\\
1 & 0
\end{array}
\right]
\oplus
\bigoplus_{j=1}^{n-1} \left[
\begin{array}{cc}
0 & 1\\
-1 & 0
\end{array}\right]\right)
\end{equation}
depending on having
even or odd parity. 
Turning to Hamiltonians, energy expectation values are computed as
\begin{equation}
    \langle\psi|H|\psi\rangle = {\rm tr}(\gamma h),
\end{equation}
with $h=-h^T$.
For a local Hamiltonian $H=\sum_i H_i$, each of the terms $H_i$ will correspond to a matrix 
    $h_i= -h_i^T$
with $\|h_i\|\leq 1$
that is a zero matrix except a $2k\times 2k$ block, since each $h_i$ acts on $k$ modes only.
The Hamiltonian matrix $h$ can without loss of generality be assumed
to be $2\times 2$ block diagonal, as any special
orthogonal transformation to bring it into this
form can be absorbed in the $O$ of the initial
covariance matrix.  
The attainable energy expectations can be computed from the reachable set 
\begin{equation}
\left\{P(O\gamma_0 O^T): O\in SO(2n)\right\},
\end{equation}
 where $P$ is the projection onto $2\times 2$
block diagonal form. By virtue of the analog of the Schur-Horn theorem for skew-symmetric matrices~\cite{leite_geometric_1999}, it becomes clear that within both the even and the odd parity sectors, the reachable set are all $2\times 2$ skew-symmetric real
block diagonal 
matrices for even and odd parity, respectively.
As a consequence of that, there is a parametrized curve 
$t\mapsto O(t)$ for $t\in[0,1]$
with $O(t)\in SO(2n)$ for all $t$ so that
\begin{equation}
    \gamma= O(0)\gamma_0 O(0)^T
\end{equation}    
and
\begin{equation}
    \omega= O(1)\gamma_0 O(1)^T
\end{equation} 
so that
\begin{equation}
    {\rm tr}(O(t)\gamma_0 O(t)^T h)
    = (1-t) {\rm tr}(\gamma h) + 
    t{\rm tr}(\omega h).
\end{equation}
That is to say, one can linearly interpolate between the
initial and final energy values. 
One can then chop the linear interpolation into a finite number $N$ steps, each of which is characterized by an orthogonal matrix in $SO(2n)$ close in operator norm to the identity. 
What is more, following the special orthogonal fermionic
analog of the decomposition of Ref.~\cite{PhysRevLett.73.58}, this transformation can be exactly decomposed into a
an $O(n^2)$ sized circuit of $2$-local fermionic
Gaussian quantum gates that are also close to the identity. 
The so obtained discrete local fermionic circuit $\prod_{i = 1}^{O(n^2) N } O_i$ therefore remains close to the continuous curve $O(t)$ for all $t \in [0,1]$.
This implies that the energy along this circuit cannot deviate too much from the initial and final value. 
By increasing the value of $N$ we can push this deviation down arbitrarily far so as to satisfy the energy constraint throughout the path, providing evidence for our conjecture. 
We leave the details of this interesting problem relevant for practical quantum 
memories with Majorana fermions for future work.

\section{Discussion}



Pinning exemplifies the mathematical question of Hamiltonian purification \cite{purification}, which we looked at here in a variety of contexts (commuting, stoquastic, permutation, and other restricted classes of Hamiltonians). 
We have presented several results in Hamiltonian complexity, raising questions about the static (complexity) and dynamical (evolution and universality) implications of a special type of control on a small subsystem. Let us now discuss a few observations.

First, quantum perturbation gadgets that have been used in Hamiltonian complexity for a long time ever since \cite{KKR06}, are also based on a form of pinning -- effectively fixing part of a system into a subspace by providing a large energy penalty to the orthogonal subspace. 
They can result in an effective Hamiltonian with multiplicatively combined, higher-locality terms, thanks to the form of the perturbative expansion of the Hamiltonian's self-energy. 
On the other hand, pinning as we view it here, is a geometrical restriction on a part of a system. First of all, it is not perturbative, and second, it can effectively generate only linear and not multiplicative combinations of operators. 
Therefore, it does not allow one to combine operators to increase the effective locality of terms, which perturbative gadgets are designed to do. 
On the contrary, we need $k+1$ local terms in a pinned Hamiltonian to get an effective $k$-local Hamiltonian.
In particular, to show that \emph{Pinned Commuting 3-LH} is \emph{\QMA} complete in Section \ref{sec:pincomm}, we have turned a $2$-Local Hamiltonian problem with promise $b,a$, into a pinned version with a doubled Hilbert space by adding a qubit. 
Moreover, the newly formed up to $3$-local and commuting terms have the form $Z$, $X$, $ZX$, $XX$, $ZZ$, $ZZX$, or $XXX$. 
However, is the increase in locality essential? The complexity of \emph{Pinned $2$-Local Commuting Hamiltonian} 
remains open. Straightforward attempts mimicking perturbation gadgets to generate effective interactions with higher locality do not work. 
Similarly, we have shown in Section~\ref{sec:pinstoq} that the \emph{Pinned Stoquastic 3-LH} is \emph{\QMA}-complete. 
However, it remains open to figure out how hard the \emph{Pinned Stoquastic 2-LH} problem is.
One way to go could be to show that 2-LH with $\pm ZZ, -XX, \pm X,\pm Z$ terms is \emph{\QMA}-complete.

Second, our reason for investigating pinning was its application to Hamiltonians with a restricted form. 
Could pinning be ``forced'' with such restricted terms?
Sometimes, as in the application to  \emph{GSCON}, there exist operators with the desired form, which energetically penalize a subspace. For example, in Section~\ref{sec:stoqGSCON}, we wrote down the stoquastic operator \eqref{R3} that works as a projector onto the complement of $\ket{-}^{\otimes 3}$ and $\ket{+}^{\otimes 3}$, or in Ref.~\cite{GMV}, where a $3$-local projector has the required form commuting with the rest of the Hamiltonian.
However, in other situations we can not do this. For example, we can not energetically prefer the state vector $\ket{-}$ of a qubit by stoquastic 
terms, 
as that would imply \emph{\QMA}-completeness of the \emph{Stoquastic LH} problem, which is considered unlikely.
Thus, we require pinning as an external condition in the  \emph{Pinned Stoquastic LH} problem.
Similarly, we added dynamical pinning based on repeated measurements in Section~\ref{sec:dynpin} as an external resource, and not directly as a part of the Hamiltonian.
Third, it would be interesting to see whether pinning for some restricted models could result in intermediate complexity (e.g., completeness for transverse Ising models), as classified in Ref.~\cite{universal}.

Fourth, as pinning fixes a particular value of a certain subsystem, one naturally asks about its relationship to postselection. What we propose in Section~\ref{sec:dynpin} is far from postselection. In our Zeno-effect constructions, the probability of even many successful projections tends to 1. Our results say that universality can arise even from this small degree of practical control. On the other hand, postselection is about being able to postselect (``choose'' the value of measurement results regardless of the low probability of the outcome). It is known that this incredibly powerful ability would increase the computational power immensely -- e.g., postselected \BQP\ becomes \PP\ \cite{postselection}.

Fifth, we hope that our investigation will shed light on and motivate further inquiries into the complexity of the variants of the original local Hamiltonian problems -- stoquastic, commuting, or with other restrictions.

Finally, we hope that dynamical pinning based on extra control (repeated measurements) of a single qubit, described in Section~\ref{sec:dynpin}, with a fixed interaction Hamiltonian of a restricted form, could be readily implemented in today's experimental settings. It is also our hope that the present work can 
substantially contribute to the growing body of
solutions to problems in Hamiltonian complexity beyond
assessing the computational complexity of approximating
ground state energies, signifying the richness of the
field.

\section{Acknowledgements}
We thank an anonymous referee for valuable comments to the early version of this paper, especially on universal evolution with stoquastic Hamiltonians. D.~N.\ 
has received funding from the People Programme (Marie Curie Actions) EU’s 7th Framework Programme under REA grant agreement No. 609427. His research has been further co-funded by the Slovak Academy of Sciences, 
as well as by the Slovak
Research and Development Agency grant QETWORK APVV-14-0878 and VEGA MAXAP 2/0173/17. D.~H. and J.~E. have been supported by the ERC (TAQ), the 
Templeton Foundation, and the DFG (EI 519/14-1,
EI 519/15-1, CRC 183). M.~S. thanks the Alexander-von-Humboldt 
Foundation for support.


\end{document}